\documentclass{article}

\usepackage{amsthm,amsmath,amsfonts,enumerate,graphicx,mathtools,amssymb}
\usepackage{authblk}
\usepackage{blindtext}
\usepackage{subfig}
\usepackage{mwe}

\usepackage[ruled,vlined]{algorithm2e}
\usepackage[noend]{algpseudocode}
\usepackage{rotating}
\makeatletter
\def\BState{\State\hskip-\ALG@thistlm}
\makeatother

\RequirePackage[colorlinks,citecolor=blue,urlcolor=blue]{hyperref}

\newcommand{\bb}{\mathbf}
\newcommand{\PP}{\mathbb{P}}
\newcommand{\EE}{\mathbb{E}}
\newcommand{\II}{\mathbb{I}}
\newcommand{\RR}{\mathbb{R}}

\DeclareMathOperator*{\argmin}{arg\,min}

\numberwithin{equation}{section}
\theoremstyle{plain}
\newtheorem{thm}{Theorem}[section]
\theoremstyle{remark}
\newtheorem{rem}{Remark}[section]

\providecommand{\keywords}[1]
{
  \small	
  \textbf{\textit{Keywords---}} #1
}

\begin{document}

\title{Interpretable random forest models through forward variable selection}

\author[1]{Jasper Velthoen}
\author[1]{Juan Juan Cai }
\author[1]{Geurt Jongbloed}

\affil[1]{Department of Applied Mathematics, Delft University of Technology, Mekelweg 4 2628 CD Delft}

\maketitle

\begin{abstract}
	Random forest is a popular prediction approach for handling high dimensional covariates. However, it often becomes infeasible to interpret the obtained high dimensional and non-parametric model. Aiming for obtaining an interpretable predictive model, we develop a forward variable selection method using the continuous ranked probability score (CRPS) as the loss function. Our stepwise procedure leads to a smallest set of variables that optimizes the CRPS risk by performing at each step a hypothesis test on a significant decrease in CRPS risk. We provide mathematical motivation for our method by proving that in population sense the method attains the optimal set. Additionally, we show that the test is consistent provided that the random forest estimator of a quantile function is consistent.

In a simulation study, we compare the performance of our method with an existing variable selection method, for different sample sizes and different correlation strength of covariates. Our method is observed to have a much lower false positive rate. We also demonstrate an application of our method to statistical post-processing of daily maximum temperature forecasts in the Netherlands.  Our method selects about 10\% covariates while retaining the same predictive power.
\end{abstract}


\section{Introduction}

In the past decades, random forests \cite{Breiman2001} have gained traction in many areas of application simply because random forests provide good predictive power. A random forest combines several trees, each obtained by recursively making axis-aligned splits in the covariate space until a stopping criterion is reached. The initial algorithm for random forests in \cite{Breiman2001} provides a good approach for conditional mean regression and classification. Later on, the approach was extended to estimate quantiles by  \cite{Meinshausen2006} and further improvements were made in \cite{Athey2019}, which introduced a quantile based splitting criterion. Due to the results in \cite{Meinshausen2006} and \cite{Athey2019}, random forests are also used for estimating the conditional quantile function. 

These quantile forests have been used in statistical post-processing to obtain probabilistic forecasts, e.g. \cite{Taillardat2016},  \cite{Taillardat2017} and  \cite{Whan2018}. Post-processing is used as a second step in weather forecasting following a  first step of physical modelling, see \cite{Kalnay2003}. This first step entails a numerical weather prediction (NWP) model that uses non-linear partial differential equations of atmospheric flow on a spatial and temporal grid. Together with parametrizations of unresolved physical processes within the grid cells and an estimated initial condition, which is obtained from observational data and a so called first guess (i.e. a forecast for that time based on a previous NWP model run), the NWP model approximates the solution to the partial differential equations. An ensemble prediction system (EPS) adds uncertainty quantification to the NWP model by computing an ensemble of forecasts for perturbed initial conditions and/or the parametrization schemes \cite{Kalnay2003}.

Generally there is still a need for bias correction and calibration of numerical weather forecasts, which motivates the second step: statistical post-processing. Historical forecasts together with the corresponding observations are used in post-processing to estimate their statistical relationship. This relationship can then be used in order to calibrate future forecasts. 

{ When post-processing forecasts of
a weather phenomenon, a better performance is often attained by adding more
information from the NWP models as predictors.} For example,  \cite{Whan2018} showed that the post-processed precipitation forecasts perform substantially better  when indices of atmospheric
instability from the NWP models are used in modelling the statistical relation. The improvement is due to the fact that the indices of atmospheric instability  help to distinguish between different types of precipitation. A full day of drizzle might accumulate to the same amount as a quick shower. However, the distributions of  precipitation under these two different weather conditions are very different. Incorporating NWP forecasts of other weather phenomena enables the model to capture such differences. 

A natural question is now: `` Which additional forecasts contain useful information on the phenomenon that one is post-processing?"  The set of potential forecasts to include in the statistical model is generally very large and furthermore they exhibit large correlations. In practice, including too many variables often leads to a decrease in statistical efficiency, and more importantly the model becomes hard to interpret. For a practitioner, it is important to understand which variables play key roles in the statistical model and how they calibrate the EPS forecast. This motivates variable selection procedures in statistical post-processing.

A random forest is generally seen as a method that deals rather well with high dimensional covariates. This property comes from the fact that in the tree fitting algorithm, a random forest chooses, the split variables  and split points, in a greedy way based on a certain criterion, e.g. the variance. This is often rather effective in the beginning of the tree fitting as many observations are split, but deep down in the tree there are fewer observations which makes the splitting criterion subject to  higher variances. Therefore global variable selection methods are considered in the literature to improve statistical efficiency and interpretation of the random forest model.

Variable selection in random forests is mainly done in terms of two types of importance measures. The first type calculates the decrease in impurity of a split made in a tree. In \cite{Louppe2013} consistency of these measures is shown on fully randomized trees. But in practice in a random forest setting these impurity measures are shown to exhibit biases (\cite{Strobl2007}). The second type is the permutation measure introduced in \cite{Breiman2001}. This measure computes how much the predictive performance decreases by randomly permuting one single predictor, which breaks the relation between response and the predictor. A popular approach is to perform a backward selection based on the permutation measures, where the model with the best predictive performance is chosen, see e.g. \cite{Genuer2010}, \cite{Fox2017} and  \cite{Gregorutti2017}. 

Correlation between predictors has a large effect on the permutation importance scores. An initial approach of dealing with this is to consider conditional importance scores, \cite{Strobl2008}. This has the downside that in some way the conditioning variable has to be chosen. A more precise analysis of the effect of correlation on permutation measures is done in \cite{Gregorutti2017}, where they conclude that a backward selection is better able to handle correlation between predictors than other strategies incorporating variable importance measures. We show in our simulation study that although the correct variables are often selected by the backward selection, there is no  control on the rate of selected noise variables, i.e. the false positives. 

In this paper, we propose a new method of selecting variables with random forests. By using the so-called continuous ranked probability score as the loss function (cf. \eqref{eq:crps}), we are able to select variables that are informative for the entire conditional distribution instead of just for the conditional mean. The procedure estimates the predictive risk based on the so-called out-of-bag samples (cf. Section \ref{outb}), which is similar to leave-one-out cross validation. Finally, we introduce a hypothesis test for each selection step to test whether a variable significantly decreases the predictive risk. We show by a detailed simulation study that our method controls the false positive rate much better than the backward selection method introduced in \cite{Gregorutti2017}, even in the presence of high correlations.

The outline of the paper is as follows. In Section \ref{sec:mathematical-set-up}, we give a detailed description of the mathematical set-up of the variable selection procedure. Then in Section \ref{sec:method}, we give a small introduction to random forests and show how the variable selection can be applied to the random forest set-up. A comparison with backward selection based on permutation measures is made in Section \ref{sec:sim}. In Section \ref{sec:data}, we apply the method to a practical example of post-processing maximum temperature forecasts and compare it to a standard method in post-processing. Finally, we end with a discussion in Section \ref{sec:dis}.

\section{Forward selection}
\label{sec:mathematical-set-up}

In this section, we describe the mathematical set-up of our forward variable selection method. We provide the intuition of the procedure together with some theoretical motivation. 
For now, we  consider a pair of random observations $(\bb X,Y)$, where $\bb X \in  \RR^p$ and $Y \in \RR$. Let $J \subset \{1, \ldots ,d\}$ denote a set of indices corresponding to the entries of the covariate vector $\bb X$ and $\bb X^J$ denote the vector with the entries from $\bb X$ corresponding to $J$. 

Let  $F_{Y|\bb X^J}$ denote the conditional distribution function of $Y$ given $\bb X^J$. And, let $L(Y,\bb X^J,F_{Y|\bb X^J})$ denote a loss functional measuring the loss between the observation $Y$, the quantity that we want to predict, and $F_{Y|\bb X^J}$,  e.g. the squared error loss $(Y - \int z\mbox{d}F_{Y|\bb X^J}(z|\bb X^J))^2$. In this section, we work from the population perspective and use exact distribution functions. 
The next section will be concerned with the estimation of the conditional distributions using random forests. 

Corresponding to the loss functional, we can now define a risk functional for the subset of variables corresponding to $J$,
\begin{align} \label{eq:R}
R(J)=\EE_{Y,\bb X}[	L(Y,\bb X^J,{F}_{Y|\bb X^J})]. 
\end{align}

In our approach an ideal variable selection procedure selects the set of variables corresponding to $J$ that minimize this risk functional. Define $m_R=\min_{J \subset \{ 1, \ldots , d\}}  R(J)$. Then, the optimal set of variables denoted by $\bb X^{J^*}$ is such that 
\begin{equation}
R(J^*)=m_R ~~\text{ and } ~~  |J^*|=\min \{|J|: R(J)=m_R\} \label{eq: J*}
\end{equation}
 where $|J|$ denotes the cardinality of $J$. 
The goal is to identify the smallest model that reaches an optimal risk. This is desirable when it comes to estimating the conditional distribution of $Y$.
It is important to note that  $J^*$ is not necessarily unique. For example two collinear covariates $X_1$ and $X_2$ both contain the same information of $Y$, then including any of the two covariates would result in the same expected loss.

In order to obtain $J^*$, one could evaluate $R(J)$ for all $2^d$ possible sets, which is often computationally infeasible. Instead we propose a forward variable selection approach as follows. We construct a sequence of length $d+1$ of nested sets $J_j$ for $j = 0, \ldots , d$ where $J_0 = \emptyset$ and
\begin{equation}\label{eq:forward-sets}
	J_j = J_{j-1} \cup \left\{  \argmin_{q \notin  J_{j-1}}  R\left(J_{j-1} \cup \{q\}\right)  \right\}.
\end{equation}
Our proposed forward selection procedure selects an optimal set  $J^o$ such that it is the smallest set attaining the minimum risk among $J_j$, $j=0,\ldots d$. More precisely, 
\begin{equation}
J^o=J_{\min\{j: R(J_j)=\min_{0\leq i\leq d}R(J_i) \}}.
\end{equation}

From this point on in the paper, we will choose the loss function equal to the Continuous Rank Probability Score (CRPS), see \cite{Gneiting2007}, defined by,

\begin{equation} \label{eq:crps}
	L(Y,\bb X^J,F_{Y|\bb X^J}) = \int_{-\infty}^{\infty} \left( I(Y \leq z) - F_{Y|\bb X^J}(z|\bb X^J)\right)^2\mbox{d}z.
\end{equation}
The CRPS compares the distribution $F_{Y|\bb X^J}$ with the ideal deterministic forecast, of which the distribution function equals the step function at the observation $Y$. The CRPS is a proper scoring rule for a large class of distribution functions; see Section 4.2 in \cite{Gneiting2007}.


In the theorem below we show that under the assumption of independent covariates, the set $J^o$ and  $J^*$ coincide.

\begin{thm} \label{thm:ordering}
Let $X_1,\ldots, X_d$ and $\epsilon$ be independent random variables. Let $h : \RR^{|J^*| + 1} \to \RR$ be a real valued measurable function and define $Y=h(\bb X^{J^*},\epsilon)$, where $J^*\subseteq \{1,\ldots,d\}$. Assume that $\EE [Y^2] <\infty$, and for any $I\subsetneq J \subseteq J^*$,  there exists a set $S \subseteq \RR$ with positive Lebesgue measure such that $\EE[I(Y\leq z)|\bb X^J]$ is not $\sigma(\bb X^I)$ measurable for all $z\in S$. Then $J^*$ is the unique subset of $\{ 1, \ldots , d\}$ satisfying \eqref{eq: J*}, and $J^0=J^*$.

	\begin{proof}
		Let $(\Omega,\mathcal{A},\mu)$ denote the probability space supporting $X_1,\ldots, X_p$ and $\epsilon$. Define the standard inner product on $L^2(\Omega,\mathcal{A},\mu)$ by  $(Z_1, Z_2)=\EE(Z_1 Z_2)$, for  any random variables $Z_1$ and $Z_2$ on $(\Omega,\mathcal{A}, \mu)$. Then $L^2(\Omega,\sigma(X_1,\ldots, X_p, \epsilon),\mu)$ becomes a Hilbert space, where the conditional expectation $\EE(Z|\bb X^J)$ is the orthogonal projection of $Z$ onto the closed linear subspace $L^2(\Omega,\sigma(\bb X^J),\mu)$. Now we have 
		\begin{align*}
		R(J)= &\int_{-\infty}^{\infty}\EE \left[ \left( I(Y\leq z) - {F}_{Y|\bb X^J}(z|\bb X^J)\right)^2\right] \mbox{d}z\\
		=&\int_{-\infty}^{\infty}\EE \left[ \left( I(Y\leq z) - \EE[I(Y\leq z)|\bb X^J]\right)^2\right] \mbox{d}z\\
		=:&\int_{-\infty}^{\infty} g_J(z)\mbox{d}z.
		\end{align*}
As the conditional expectation equals the orthogonal projection, for any $z\in \RR$,

\begin{equation}
	g_J(z)=\min_{G\in \sigma(X^J)}\EE[(I(Y\leq z)-G)^2].
\end{equation}	
Therefore, for any $z\in \RR$, if $J_1 \subset J_2$, we have

\begin{equation} \label{eq:monotonicity}
	g_{J_1}(z)\geq g_{J_2}(z).
\end{equation} 
This implies that $R(J_1) \geq R(J_2)$.		

Next, note that if $J_2=J_1\cup \{j\}$ and $j\notin J^*$, then for any $z\in \RR$,
\begin{equation}
g_{J_1}(z)=g_{J_2}(z).
\end{equation}

This is because $\EE[I(Y\leq z)|\bb X^{J_2}]=\EE[I(Y\leq z)|\bb X^{J_1}]$ by the independence of $Y$ and $X_j$. In this case $R(J_1)=R(J_2)$.

Finally, we show that if $J_2=J_1\cup \{j\}$, where $j \in J^*$ then $R(J_1) > R(J_2)$. We prove by contradiction. If not, then $R(J_1) = R(J_2)$, which means in view of \eqref{eq:monotonicity} that 
$g_{J_1}(z)=g_{J_2}(z)$, for all  $z\in \RR\setminus C$, where $C$ has zero Lebesgue measure.

From here we denote $I(Y\leq z)$ by $I_z$ to simplify notation. Expanding the squares and using the tower property of conditional expectation we see that
\begin{align*}
	g_{J_1}(z) &- g_{J_2}(z) =
	\EE \left[ \left( I_z - \EE[I_z|\bb X^{J_1} ]\right)^2 - \left( I_z - \EE[I_z|\bb X^{J_2}] \right)^2 \right]\\
	&= \EE \left[ -2 I_z\EE[I_z|\bb X^{J_1} ]  + \EE[I_z|\bb X^{J_1} ]^2 +  2 I_z\EE[I_z|\bb X^{J_2} ]  - \EE[I_z|\bb X^{J_2} ]^2 \right]\\
	&= \EE \left[ \EE \left[ -2 I_z\EE[I_z|\bb X^{J_1} ]  + \EE[I_z|\bb X^{J_1} ]^2 +  2 I_z\EE[I_z|\bb X^{J_2} ]  - \EE[I_z|\bb X^{J_2} ]^2\right| \left. \bb X^{J_2} \right] \right]\\
	&= \EE \left[ (\EE[I_z|\bb X^{J_1}] -  \EE[I_z|\bb X^{J_2}])^2 \right] = 0.
\end{align*} 

From this we conclude that $\EE[I_z|\bb X^{J_1}] = \EE[I_z|\bb X^{J_2}]$ for all $z\in \RR\setminus C$. This implies that $\EE[I_z|\bb X^{J_2}]$ is $\sigma(X_{J_1})$ measurable which contradicts our assumption, hence $R(J_1) > R(J_2)$. 

We can now observe that the forward sets are built by adding variables from $J^*$ until all variables of $J^*$ have been added, therefore $J^0 = J^*$. 
	\end{proof}
\end{thm}

\begin{rem}
	The assumption : $\EE[I(Y\leq z)|\bb X^J]$ is not $\sigma(\bb X^I)$-measurable for any $I\subsetneq J \subseteq J^*$, is used to prove the uniqueness of  $J^*$. As we know that $R(J^*) = R(J^* \cup \{ j \})$ for $j \notin J^*$, there are many sets, which have minimal risk in population sense. The assumption essentially ensures that $J^*$ does contain only indices $j$ such that the function $h$ is not constant for $x_j$ almost everywhere with respect of the distribution of $X$.
\end{rem}

\begin{rem}
	The choice of the CRPS loss function is motivated by our application. Though for different loss functions $L$ that focus on a specific part of the conditional distribution, the procedure explained in this section could still be applied.
\end{rem}

\section{Forward selection using random forests}
\label{sec:method}

We use a random forest to estimate the conditional distribution function ${F}_{Y|\bb X^J}$ and the risk. Now, we make a little  excursion to explain the random forest algorithm. We follow the tree construction algorithm proposed in \cite{Wager2018} and the extension for quantile estimation from \cite{Athey2019}. We choose this approach because it is the only approach that makes splits based on a quantile criterion, additionally in \cite{Athey2019} asymptotic normality  for the quantile estimates is established.

\subsection{Intermezzo: Random Forests}
\label{Intermezzo: Random Forests}

Denote the data set by $(\bb X_1,Y_1), \ldots ,(\bb X_n,Y_n)$. A random forest is defined as a collection of trees. Each tree $T$ is obtained by recursively splitting a set of observations by making axis-aligned splits in the covariate space, meaning a split is made on a single covariate value at a time. As a result, every tree induces a partitioning of the covariate space in possibly semi-infinite hyper rectangles. Denote the conditional quantile function by $Q_{Y|\bb X}(\tau|\cdot)$, where $\tau \in (0,1)$  denotes a probability level. In this section we focus on fitting a forest in order to estimate the function $Q_{Y|\bb X}(\tau|\cdot )$. The estimation procedure for $Q_{Y|\bb X^{J}}(\tau|\bb X^{J})$ works exactly the same by  fitting a forest based on $\{(\bb X^{J}_1,Y_1), \ldots ,(\bb X^{J}_n,Y_n) \}$.

Recurrent splits are made starting with parent node $P$, a node in the current partition, creating two child nodes $C_1$ and $C_2$, such that $P = C_1 \cup C_2$ and $C_1 \cap C_2 = \emptyset$. This split should be informative with respect to $Q_{Y|\bb X}(\tau|\cdot)$ and is chosen to maximize,
\begin{equation} \label{eq:split-crit}
	e(C_1,C_2) = \frac{n_{C_1}n_{C_2}}{n_P} (Q_{Y|\bb X}(\tau|\bb X \in C_1) - Q_{Y|\bb X}(\tau|\bb X \in C_2))^2,
\end{equation}
where $n_P$, $n_{C_1}$, $n_{C_2}$ are the number of observations $\bb X_i$ in each node. In practice this makes the the algorithm very slow as it requires the computation of two quantiles for each possible split. Instead in \cite{Athey2019} a relabelling step is proposed and defined as $\II(Y_i > Q_{Y|\bb X}(\tau|\bb X \in P))$ for the $\tau$ quantile. Now a standard regression split, as used in a standard random forest \cite{Breiman2001}, is made on the labels. This means to maximize the squared difference between the average label in both child nodes.


The trees fitted in \cite{Wager2018} and \cite{Athey2019} are called honest trees and are slightly different from the standard structure of tree fitting. A tree is fit by first sub-sampling a set of indices from $\{1,\ldots ,n\}$ of size $s << n$ and then randomly splitting this sub-sample in two sets $\mathcal{I}$ and $\mathcal{J}$ both of size $s/2$ each. Recursive splits of $\RR^d$ are then made based on criterion \eqref{eq:split-crit}, with data points $(Y_i,\bb X_i) : i \in \mathcal{I}$. The tree becomes honest by  removing all the data points indexed by set $\mathcal{I}$ and using only the data points indexed by set $\mathcal{J}$ for estimation of $Q_{Y|\bb X}(\tau|\bb x)$ for a new observations $\bb X$. 

A random forest is then obtained by fitting $B$ trees. Denote by $l_b(\bb X)$ the leaf node of tree $b$ in which $\bb X$ falls. Then for $1\leq i\leq n$, the weight for $(\bb X_i, Y_i)$ induced by the $b$th tree is given by,

\begin{equation}
	w_{i,b}(\bb X) = \frac{\II( i \in \mathcal{J} ~~ \& ~~\bb X_i \in l_b(\bb X))}{\sum_{j\in \mathcal{J}} \II( \bb X_j \in l_b(\bb X))}, 
\end{equation}
where $\frac{0}{0}=0$.
The forest weights are obtained by averaging the tree weights over the $B$ trees, $w_i(\bb x) = \frac{1}{B}\sum_{b=1}^B w_{i,b}(\bb X)$. An estimate of $\hat{Q}_{Y|\bb X}$ is then given by the locally weighted estimated quantile,

\begin{equation}\label{eq:quant_reg}
	\hat{Q}_{Y|\bb X}(\tau|\bb x) = \argmin_{\theta} \sum_{i=1}^n w_i(\bb x)\rho_{\tau}(Y_i - \theta),
\end{equation}
with $\rho_{\tau}(u) = u(\tau - \II(u < 0))$ the quantile check function. Note that the structure is similar to kernel regression, but instead of a deterministic kernel with bandwidth $h$ the weights are determined by the data via the forest. Random forests are sometimes called adaptive nearest neighbour estimators for this reason.

In the variable selection procedure we aim to select variables that are predictive for the conditional distribution. Therefore, instead of building random forests with respect to a single $\tau$ quantile, consider a sequence of quantiles $0 < \tau_1, \ldots, \tau_K < 1$. This needs a different type of relabelling than for a single quantile as explained above. They define the relabelling then by,

\begin{equation*} \label{eq:relabeling}
	Z_i = \sum_{k=1}^{K} I( Y_i \leq \hat{Q}_{Y|\bb X}(\tau_k|\bb X \in P)). 
\end{equation*}
The best split is then chosen to maximize the following multi class classification rule:
\begin{equation*} \label{eq:multi-classification}
	\hat{e}(C_1,C_2) = \frac{\sum_{k=1}^{K}\left[ \sum_{X_i \in C_1} I(Z_i = k)\right]^2}{n_{C_1}} + \frac{\sum_{k=1}^{K}\left[ \sum_{X_i \in C_2} I(Z_i = k)\right]^2}{n_{C_2}}.
\end{equation*}

\subsection{Estimation of predictive loss} \label{outb}
The main quantity in the theoretical framework from Section \ref{sec:mathematical-set-up} is the CRPS risk. To make use of the random forest quantile estimator, we use an equivalent expression of the CRPS loss (\eqref{eq:crps}), that is, 
$\int_{-\infty}^{\infty} \left( I(Y \leq z) - F_{Y|\bb X^J}(z|\bb X^J)\right)^2\mbox{d}z= 2\int_{0}^{1} \rho_{\tau}(Y - {Q}_{Y|\bb X^{J}}(\tau|\bb X^{J})) \mbox{d}\tau$.
The equivalence of these two definitions is shown in the appendix. Plugging in the estimated quantile function, we obtain the  following targeted loss in the estimation context: 
\begin{equation}\label{eq:crps-quant}
	L(Y,\bb X^{J}, \hat{Q}_{Y|\bb X^{J}})  = 2\int_{0}^{1} \rho_{\tau}(Y - \hat{Q}_{Y|\bb X^{J}}(\tau|\bb X^{J})) \mbox{d}\tau.
\end{equation}
Here we denote $\hat{Q}_{Y|\bb X^{J}}$ as the random forest estimator of the conditional quantile function with respect to the dataset $\{(\bb X^{J}_1,Y_1), \ldots ,(\bb X^{J}_n,Y_n) \}$ and with two arguments, a probability level $\tau$ and the covariate vector $\bb X^J$.


 A naive way to estimate the expected loss (that is the expectation of \eqref{eq:crps-quant}), would be considering 
$$
\frac{2}{n}\sum_{i=1}^n\int_{0}^{1} \rho_{\tau}(Y_i - \hat{Q}_{Y|\bb X^{J}}(\tau|\bb X_i^{J})) \mbox{d}\tau.
$$
However, this would lead to over-fitting because the training set (data for estimating ${Q}_{Y|\bb X^{J}}$) are the same as the testing set (data for estimating the expectation). This problem can be circumvented  by using so called out-of-bag samples as test set.   

The out-of-bag samples for the $b$th tree are defined as the samples that are not used for generating the tree. For each observation $(\bb X^{J}_i,Y_i)$, a sub forest $\mathcal{F}_i $ is defined by $\mathcal{F}_i = \{ T_b : i \notin (\mathcal{I}_b \cup \mathcal{J}_b) \}$. Namely, this sub forest consists of trees for which $(\bb X^J_i,Y_i)$ is out-of-bag.  Observe that the number of trees in $\mathcal{F}_i $ is random and hence not necessarily the same for all $i$. The expected number of trees for each sub forest is $B\left(1-\frac{s}{n}\right)$.



We use the sub forest $\mathcal{F}_i$ to estimate the conditional quantile function and denote it with $\hat{Q}_{Y|\bb X}^{\mathcal{F}_i}(\tau|\bb X^J)$. Since the trees in sub forest $\mathcal{F}_i$ do not use observation $(\bb X^{J}_i,Y_i)$, we use this quantile estimator to evaluate the CRPS loss for $(\bb X^{J}_i,Y_i)$.
Doing this for all observations, we obtain the estimated CRPS risk given by,  
\begin{equation}\label{eq:pred-risk-estimation}
	\hat{R} (J):= \frac{2}{n} \sum_{i=1}^n \int_{0}^{1} \rho_{\tau}(Y_i - \hat{Q}^{\mathcal{F}_i}_{Y|\bb X^J}(\tau|\bb X_i^J)) \mbox{d}\tau.
\end{equation}
In the sequel, we write $\hat{Q}^{\mathcal{F}_i}_{Y|\bb X^J}(\tau|\bb X_i^J)=\hat{Q}^{\mathcal{F}_i}(\tau|\bb X_i^J)$ for simplicity.

This out-of-bag procedure for estimating risk  has similarities to leave-one-out cross validation. For validating the $i$th observations we use all trees which do not use the $i$th observation. The difference is that sub forests have in expectation the same size, but not exactly. Computationally the out-of-bag sample approach is also much faster compared to leave-one-out cross validation. Note that a tree has $n-s$ out-of-bag samples and hence the tree is used is used in $n-s$ sub-forests. On the other hand leave one out cross validation does not reuse trees and estimates a new forest for each element in the summation of \eqref{eq:pred-risk-estimation}.


\subsection{One step forward} \label{subsec:forward}
The forward variable selection  sequentially adds variables such that the predictive loss is minimized. We  here explain how each step is performed.
Recall that for a index set $J$, the estimated risk $	\hat{R} (J)$ is given by \eqref{eq:pred-risk-estimation}. 
Suppose that we have selected the first $j-1$ variables with indices in $\hat{J}_{j-1}$. Then the $j$th variable $X_{\hat{i}_j}$ is selected based on
\begin{equation}\label{eq:set-estimation}
	\hat{i}_{j} =   \argmin_{q \notin \hat{J}_{j-1}} \hat{R}(\hat{J}_{j-1} \cup \{ q\}). 
\end{equation}
and $\hat{J}_j = \hat{J}_{j-1} \cup \{ \hat{i}_j \}$. The procedure of a single step forward is detailed in Algorithm \ref{alg:one-step}.

\begin{algorithm}
\SetAlgoLined
\KwResult{$\hat{i}$, $\hat{R}(J \cup \{ q\}) : q \notin J$}
	Define data $(Y_1,\bb X_1), \ldots (Y_n,\bb X_n)$\;
	Define set $J \subset \{1,\ldots p\}$\;
 \For{ $q \notin J$ }{
	construct a forest with $(Y,\bb X^{J \cup \{q\}})$\;
	Calculate $\hat{R}(J \cup \{q\})$\;
 } 
$\hat{i} = \argmin_{q} \hat R(J \cup \{q\})$
\caption{\label{alg:one-step} A forward step}
\end{algorithm}

\subsection{Stopping on time}
Motivated by the result in Theorem \ref{thm:ordering}, we stop selecting variables when there is no further decrease in CRPS risk. From the proof of Theorem \ref{thm:ordering}, adding variables that are not in $J^*$ does not have an effect on the CRPS risk. In practice, where we are working with finite samples, additional covariates decrease in fact the statistical efficiency of the random forest which leads to higher CRPS values. Because of the random component in the forest procedure, different forests will have different risk. In general this can be avoided by fitting an enormous number of trees to reduce the random component, but in practice this is infeasible. Instead we use the randomness to test the following hypothesis at each step,


\begin{align*}
	H_0: &\mbox{   } R(\hat{J}_{j-1}) - R(\hat{J}_{j}) = 0\\
	H_A: &\mbox{   } R(\hat{J}_{j-1}) - R(\hat{J}_{j})  > 0.
\end{align*}

The fitted forests at $j$-th and $(j+1)$-th steps are used to obtain several estimates of $R(\hat{J}_{j-1}) - R(\hat{J}_{j})$. More precisely, we estimate this difference by $\hat{R}(\hat{J}_{j-1} \cup \{ q \})-\hat{R}(\hat{J}_{j} \cup \{ q \})$, where $q \notin \hat{J}_{j}$. Note that $\hat{R}(\hat{J}_{j-1} \cup \{ q \})$ is computed at the $j$-th step for identifying 	$\hat{i}_{j}$ and $\hat{R}(\hat{J}_{j} \cup \{ q \})$ at the $(j+1)$-th  step for identifying $\hat{i}_{j+1}$. So, the testing procedure does not require any extra forest fitting. We propose the following test statistics:

\begin{equation} \label{eq:test-statistic}
	W_q = \sum_{q \notin \hat{J}_j}  \II(\hat{R}(\hat{J}_{j-1} \cup \{q\}) - \hat{R}(\hat{J}_j \cup \{q\}) > 0 ).
\end{equation}
Under the null hypothesis, the variable added on the $j$th step does not contribute to the predictive performance of the model. As a result both risks are asymptotically equal (see the proof for Theorem \ref{thm:consistency_test}), meaning that the test-statistic approximately has a binomial distribution, $\mbox{Bin}(M_j, 0.5)$, where $M_j=d-|\hat J_j|$. We  reject $H_0$ if $W > C^j_{1-\alpha}$, where $C^j_{1-\alpha}$ is the $1-\alpha$ quantile of $\mbox{Bin}(M_j, 0.5)$. The consistency of this test is established in the theorem below.

\begin{thm} \label{thm:consistency_test}
Assume that for any $\tau \in (0, 1)$, as $n\rightarrow \infty$,
\begin{equation}
\frac{1}{n}\sum_{i=1}^n\left| \hat{Q}^{\mathcal{F}_i}(\tau|\bb X_i^J) -Q(\tau|\bb X_i^J)\right| \overset{p}{\rightarrow} 0, \label{eq:consistency_Q}
\end{equation}
where $J=\hat{J}_{j}  \cup \{q\}$ or $J=\hat{J}_{j-1}  \cup \{q\}$, $q \notin \hat{J}_j$. Then, under the assumptions of Theorem \ref{thm:ordering},
\begin{equation}
\PP(W>C^j_{1-\alpha}) \rightarrow 1 ~ ~\mbox{Under hypothesis }H_A,
\end{equation}
as $n\rightarrow \infty$.
\end{thm}

\begin{proof}
It suffices to prove that under $H_A$, as $n\rightarrow \infty$
$$
\EE[W]\rightarrow \omega_0,
$$
where $\omega_0>C_{1-\alpha}^j$.\\
Denote  $I_q:=\hat{J}_{j-1} \cup \{q\}$ and $K_q:=\hat{J}_{j} \cup \{q\}$. Then, we have 
\begin{align*}
&\hat R(I_q)-\hat R(K_q)\\
=&\frac{2}{n} \sum_{i=1}^n \left(\int_{0}^{1} \rho_{\tau}(Y_i - \hat{Q}^{\mathcal{F}_i}(\tau|\bb X_i^{I_q})) \mbox{d}\tau
-\int_{0}^{1} \rho_{\tau}(Y_i - \hat{Q}^{\mathcal{F}_i}(\tau|\bb X_i^{K_q})) \mbox{d}\tau\right)\\
=&\frac{2}{n} \sum_{i=1}^n \left(\int_{0}^{1} \rho_{\tau}(Y_i - \hat{Q}^{\mathcal{F}_i}(\tau|\bb X_i^{I_q})) \mbox{d}\tau
-\int_{0}^{1} \rho_{\tau}(Y_i - Q(\tau|\bb X_i^{I_q})) \mbox{d}\tau\right)\\
&+\frac{2}{n} \sum_{i=1}^n \left(\int_{0}^{1} \rho_{\tau}(Y_i - Q(\tau|\bb X_i^{K_q})) \mbox{d}\tau
-\int_{0}^{1} \rho_{\tau}(Y_i -\hat{Q}^{\mathcal{F}_i}(\tau|\bb X_i^{K_q})) \mbox{d}\tau\right)\\
&+\frac{2}{n} \sum_{i=1}^n \left(\int_{0}^{1} \rho_{\tau}(Y_i - Q(\tau|\bb X_i^{I_q})) \mbox{d}\tau
-\int_{0}^{1} \rho_{\tau}(Y_i - Q(\tau|\bb X_i^{K_q})) \mbox{d}\tau\right)\\
=:&S_1+S_2+S_3.
\end{align*} 
Applying the Knight's identity, $\rho_\tau(u-v)-\rho_\tau(u)=-v(\tau-I(u<0))+\int_0^v(I(u\leq s)-I(u\leq 0))\mbox{d}s$, which implies that $|\rho_\tau(u-v)-\rho_\tau(u)|\leq 2 |v|$, we have 
\begin{align*}
|S_1|\leq \frac{4}{n}\sum_{i=1}^n\int_0^1 \left| \hat{Q}^{\mathcal{F}_i}(\tau|\bb X_i^{I_q}) -Q^{\mathcal{F}_i}(\tau|\bb X_i^{I_q})\right| \mbox{d}\tau \overset{p}{\rightarrow} 0,
\end{align*}
by \eqref{eq:consistency_Q}. The same result holds for $S_2$.\\
Observe that $S_3$ is the sample mean of I.I.D. random variables with expectation $R(I_q)-R(K_q)$. Applying law of large number, $S_3\overset{p}{\rightarrow} R(I_q)-R(K_q)$. Combing with the results for $S_1$ and $S_2$, we have 
\begin{align*}
\hat R(I_q)-\hat R(K_q)\overset{p}{\rightarrow} R(I_q)-R(K_q).
\end{align*}
Under $H_a$, $\hat i_j \in J^*$, thus, by the proof for Theorem \ref{thm:ordering}, for all $q \notin \hat{J}_{j}$, 
$$
R(I_q)-R(K_q)>0.  \footnote{{ Obviously under $H_0$, $R(I_q)-R(K_q)=0$.}}
$$
This implies that 
\begin{align*}
\EE[W]=\sum_{q \notin \hat{J}_j}  \PP\left(\hat R(I_q) - \hat{R}(K_q) > 0 \right)\overset{p}{\rightarrow} M_j >B_{1-\alpha}.
\end{align*}
\end{proof}

In practice, the integration in \eqref{eq:pred-risk-estimation} is numerically approximated. Let $\tau_t = \frac{t}{k+1}, t = 1,\ldots , k$, where $k$ is 
a pre-specified integer. The estimated risk $\hat R(J)$ in \eqref{eq:pred-risk-estimation} is approximated by 
\begin{equation}\label{eq:crps-estimation}
\hat R(J) =  \frac{2}{k} \sum_{t = 1}^k \rho_{\tau_t}(Y_i - \hat{Q}_{Y|\bb X}(\tau_t|\bb X)).
\end{equation}

The complete procedure is given in Algorithm \ref{alg:FVS}.

\begin{algorithm}
\SetAlgoLined
\KwResult{$J^o$}
	Set data $(Y_1,\bb X_1), \ldots (Y_n,\bb X_n)$\;
	Set $j = 1$, $J_0 = \emptyset$, $\alpha$\;
	Calculate $J_1$ with Algorithm \ref{alg:one-step} using $J = J_0$\;
	\Repeat{$W_j \leq C_{1-\alpha}^j$ | $j == p$}{
		$j=j+1$\;
		Calculate $J_j$ with Algorithm \ref{alg:one-step} using $J = J_{j-1}$\;
	 	Calculate $W_j$ from equation \ref{eq:test-statistic}\;
	 }
	 $J^o = J_{j-1}$\;
\caption{\label{alg:FVS} Forward variable selection}
\end{algorithm}

\section{Comparison based on simulation}
\label{sec:sim}

In this section we compare the performance of our variable selection procedure with the backward selection based on a permutation measure with a mean squared error criterion proposed in \cite{Gregorutti2017}; details of the method are stated later in this section. We compare with this method as it is currently the only method that deals with correlated predictors for random forests and we will refer to it as the backMSE method. For the comparison we simulate data from the following model,
\begin{equation}
	Y = \mu(\bb X) + \sigma(\bb X)\epsilon,
\end{equation}
where $\epsilon$ follows a standard normal distribution and independent of this, $\bb X \in \RR^{25}$ follows a multivariate normal distribution. For the covariance structure of $\bb X$ we split up the covariates into blocks $I_l = \{ (l-1)*5 + 1, \ldots ,(l-1)*5 +5 \}$ for $l=1,\ldots , 5$. The covariance function of $\bb X$ is then given by,

\begin{equation}
	\mbox{Cov}(X_{j},X_{i}) = 
	\begin{cases}
		1, &\mbox{ if } i = j;\\
		\rho, &\mbox { if } i,j \in I_l \mbox{ for the same } l;\\
		0, & \mbox{ otherwise.}
	\end{cases}
\end{equation}

The two selection methods are compared for $\rho \in \{0,0.4,0.8\}$. For the functions $\mu$ and $\sigma$  three different models are considered:
\begin{align*}
	&\mu_1(\bb X) = X_1 + \frac{X_6}{2} + \frac{X_{11}}{4}, 
	&\sigma_1(\bb X)& = 1;\\
	&\mu_2(\bb X) = X_1, 
	&\sigma_2(\bb X)& = \exp\left( \frac{X_6}{2} + \frac{X_{11}}{3} \right);\\
	\mbox{and  }
	&\mu_3(\bb X) =
	\begin{cases} 
		X_6^2, \mbox{ if } X_1 \geq 0,\\
		-X_6, \mbox{ if } X_1 < 0,
	\end{cases}
	&\sigma_3(\bb X)& = 
	\begin{cases} 
		2, \mbox{ if } X_{11} \geq 0,\\
		1, \mbox{ if } X_{11} < 0.
	\end{cases}
\end{align*}
The first model is a model where the covariates only influence the mean, in the second model the influence is mainly on the variance. The third model considers discontinuous covariate dependence in both mean and variance. Finally, we choose sample sizes $n \in \{500,1000,2500\}$. 

The backMSE method evaluates the relevance of a covariate by its permutation importance measure, which is defined as 
$$
I(X_j)=\EE\left[(Y-\EE(Y|\bb X_{(j)}))^2\right]-\EE\left[(Y-\EE(Y|\bb X))^2\right], 
$$
where $\bb X_{(j)}=(X_1,\ldots,X'_j,\ldots,X_d)$ such that $X'_j=^d X_j$ and $X'_j$ is independent of $Y$ and of the other covariates.  A large score of $I(X_j)$ indicates that covariate $X_j$ is 
important. The method randomly permutes the values of $X_j$ to mimic a random sample of $X'_j$.  An estimator of  $I(X_j)$ using out of bag samples is given in (2.1) in \cite{Gregorutti2017}.


In \cite{Gregorutti2017} it is shown that the order of the permutation importance measures can not be naturally interpreted in the presence of correlation between the covariates, as variables that are correlated share their importance. As a result, the importance of the important variables is lower than it should be. The backMSE deals with this problem by iteratively removing the least important variable and refitting the model and calculating the importance scores. This process is repeated until no variables are left. The optimal model is then chosen as the model that minimizes the out-of-bag mean squared error. Why this works is easily seen with two highly correlated informative variables. Initially they do not seem important because they share their importance, but by removing one the importance is not shared any more. The left over variable shows the true importance and will therefore be in the selected set. 

It is recommended in \cite{Gregorutti2017} to compute several forests and take averages to stabilize the variable importance scores and the error estimates. We compute for each step 20 forests where each forest contains $2000$ trees. For this method, we follow the standard forest algorithm from \cite{Breiman2001}, fitting trees based on bootstrap samples of size $n$, $mtry$ is set to the default value for regression $p/3$ and taking a minimum leaf size of $5$.

For our method we also take fixed parameters with sub sampling fraction $s = 0.5$, a minimum node size of $1$, $mtry=p$ and $1000$ trees. We have tested the influence of these tuning parameters on several simulation models and the results are rather robust to different choices. Our selection model adds variables one at a time and stops when additional variables do not increase performance. As the model is therefore often small it makes sense to not over randomize by setting $mtry$ to smaller than $p$.  We advise to choose a small $s$ for large datasets in order to reduce computation time.

\begin{sidewaysfigure}
    \centering
    \subfloat[Model 1]{ \label{fig:model1} {\includegraphics[width=0.4\textwidth]{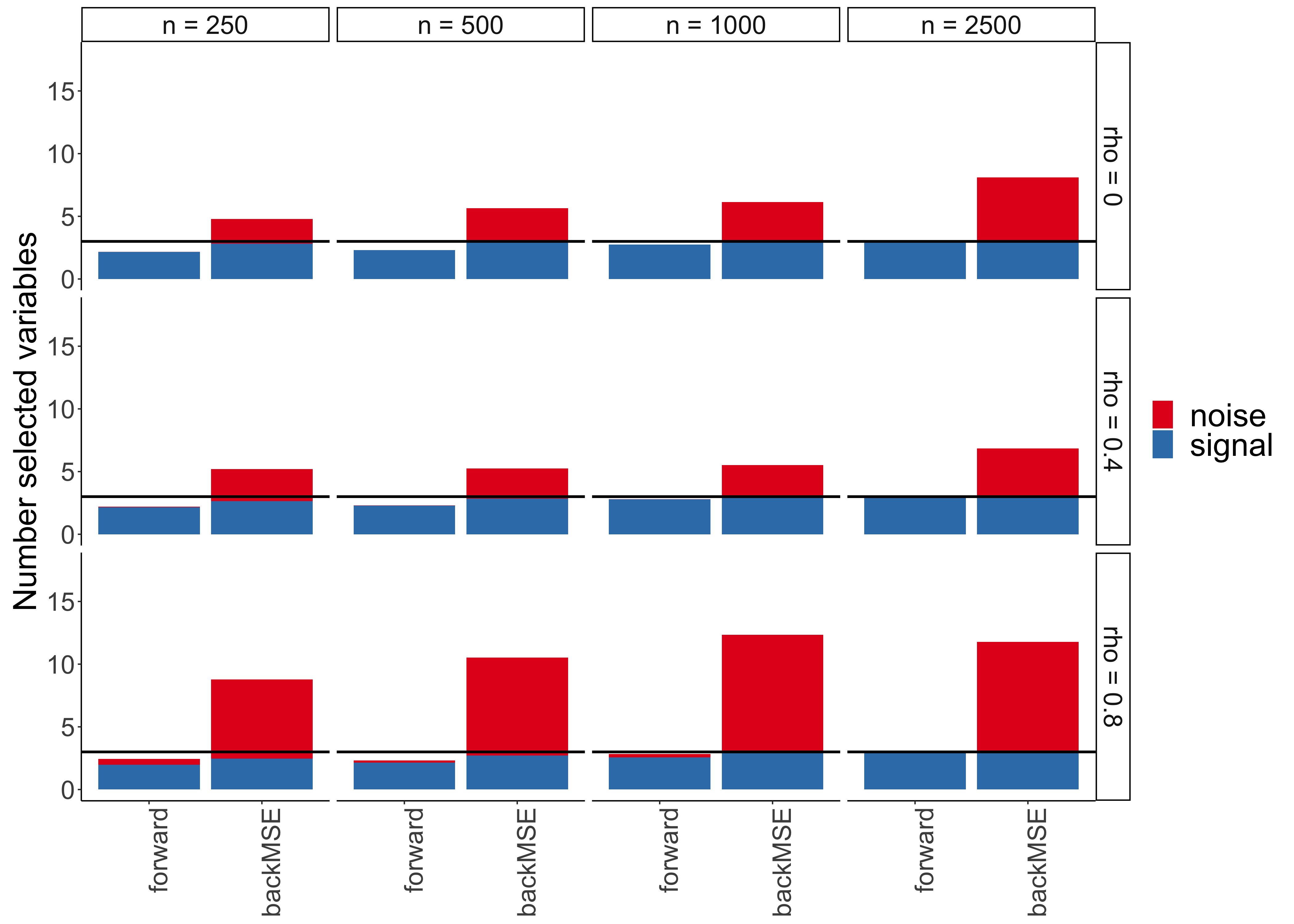} }}%
    \subfloat[Model 2]{ \label{fig:model2} {\includegraphics[width=0.4\textwidth]{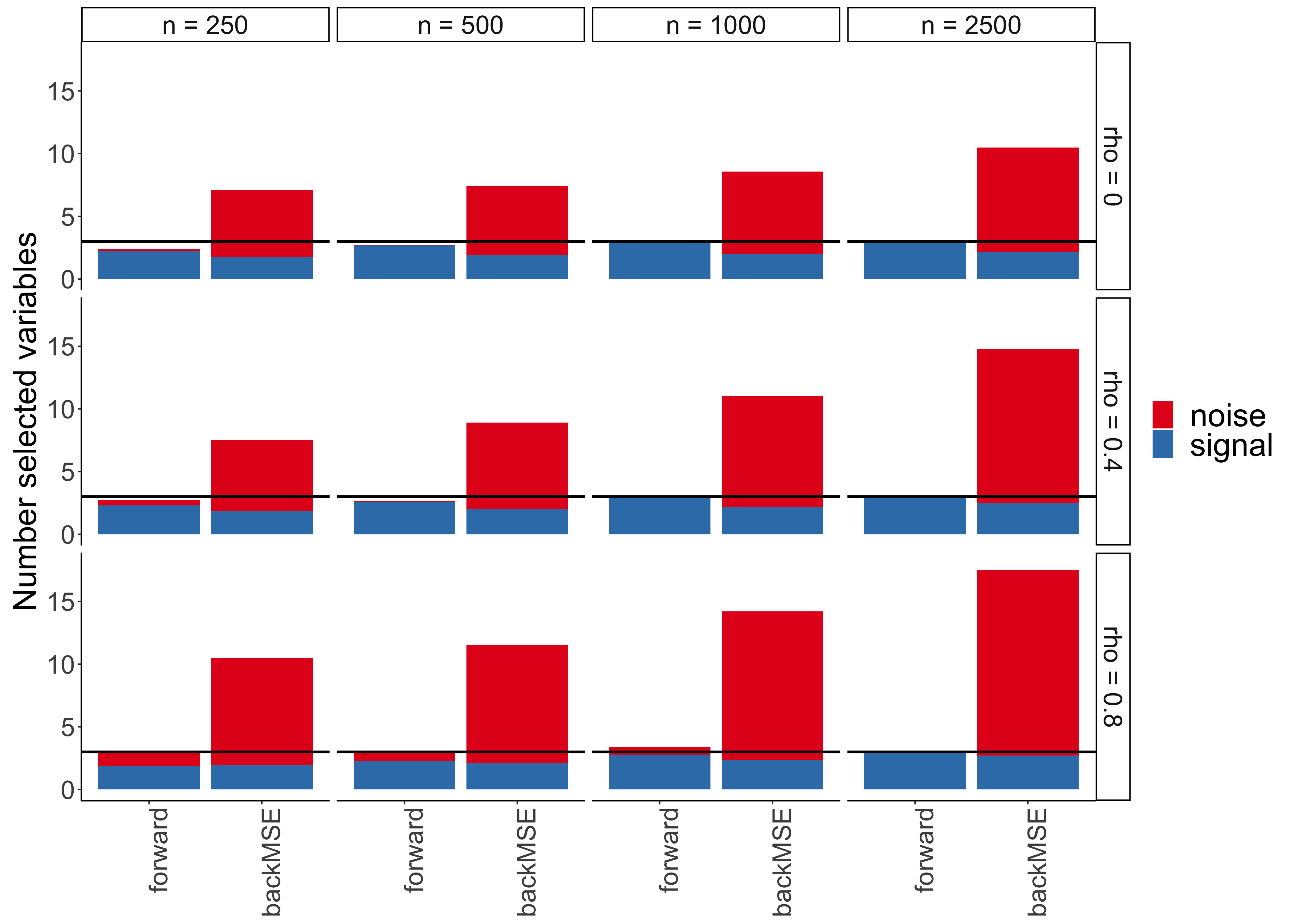} }}%
    \quad
    \subfloat[Model 3]{ \label{fig:model3} {\includegraphics[width=0.4\textwidth]{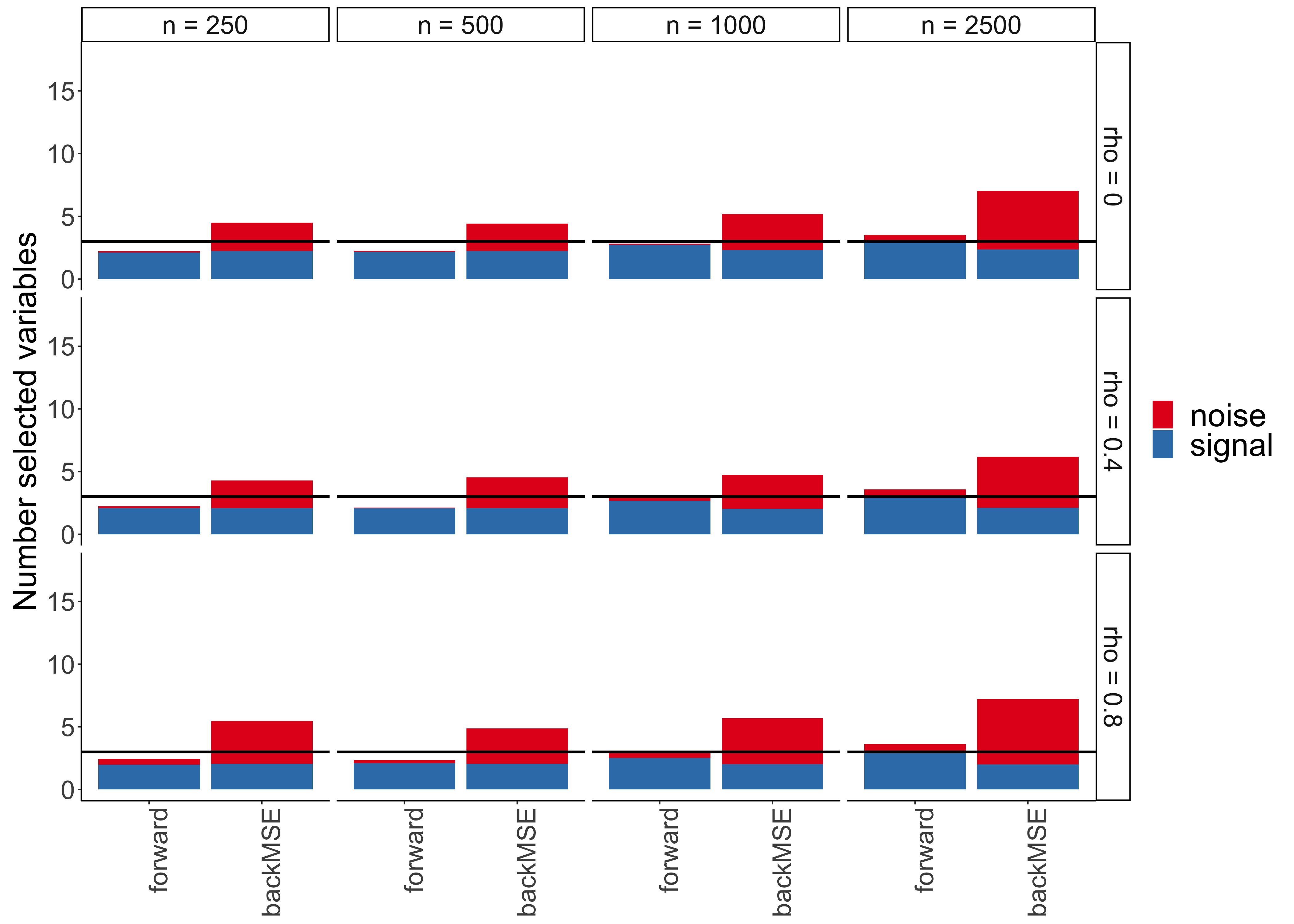} }}%
    \caption{Average number of selected variables over 100 simulations. Total number of variables in height of the bar blue for signals variables and red for noise variables. Bars from left to right correspond to the sample sizes and top to bottom corresponds to different levels of correlation. The horizontal line coincides  the total number of signal variables in the model.}
		\label{fig:model123}
\end{sidewaysfigure}
For each model we simulate 100 data sets. The results are summarized in Figure  \ref{fig:model123}. 
For the first model we see that the backMSE method retrieves more signal variables than the forward selection for low sample sizes and that as the sample size grows the forward variable selection also recovers all signal variables. A large difference is seen in the number of noise variables that are selected. The forward selection performs much better in this than the backMSE, which systematically selects noise variables and tends to even select more as the sample size increases. This phenomenon is also visible for Models 2 and 3  as  seen in Figures \ref{fig:model2} and \ref{fig:model3}. For these two models where the variance is dependent on covariates, the CRPS criterion clearly has an edge over the backMSE that selects variables based on the mean squared error and therefore has a hard time selecting these variables. 

The reason why the backMSE selects many noise variables is two-fold. First the backMSE method selects the optimal set based on a predictive mean squared error criterion. This approach does not account for the inherent variable selection within the random forest, where at each node the split that reduces the variance the most is chosen. As a result the random forest is able to ignore noise variables partially. In practice this means that in an out-of-bag performance measure the addition of a single noise variable cannot be detected. Therefore the variables that are selected will not be the smallest set, but instead a set with maximum number of noise variables maintain the lowest performance. Secondly, the backMSE does not adequately deal with the correlation. For example in Model 1 with $\rho = 0.8$ all variables $X_1,\ldots , X_5$ have higher variable importance compared to $X_6$, which means that if $X_6$ is in the model, so are $X_2,\ldots, X_5$. 

Thanks to our testing approach, a small number of noise variables is selected with the forward selection. Using the randomness induced by the random forest, our testing procedure selects a variable that leads to a significant reduction of the predictive loss. The significance level naturally controls the number of selected noise variables by the nature of the testing procedure. We have set the significance level to $5\%$ for all simulations in the paper.

\section{Post-Processing maximum temperature forecasts}
\label{sec:data}

There are substantial risks related to extremely high temperatures. Consecutive days of high temperatures, i.e. heat waves, lead to higher mortality, especially older people. Besides high temperatures can cause train rails to expand and thereby potentially disrupt the train system. Additionally, in the absence of rain they likely cause severe droughts as seen in 2018 in The Netherlands, which has had large consequences for nature areas and agriculture. The Royal Netherlands Meteorological  Institute (KNMI) issues alarms for persistent warm weather. To design a good alarm system it is essential to have good quality weather forecasts. One of the most used ensemble models, the European Centre for Medium-Range Weather Forecasts (ECMWF) ensemble model, has a negative bias in the maximum temperature forecast. As an illustration,  Figure \ref{fig:bias} shows the forecast bias for data observed at weather station de Bilt where KNMI is located. For accurate forecasts, this bias needs to be corrected for. This can be easily done by estimating the linear relation between the forecasts and the observations. Although this quickly improves the maximum temperature forecast, this leaves unused a vast amount of forecast data for other weather types. We will show that using a wide range of potential covariates, the maximum temperature forecasts are improved further than by a simple bias correction. By performing the variable selection we then also investigate in more detail what effect different covariates have on the forecast distribution estimated using the random forest model.

\begin{figure}	
	\centering
	\caption{\label{fig:bias} Scatter plot of error of the ECMWF high resolution deterministic run maximum temperature forecast against that deterministic forecast in the warm half year for the years 2011-2019 at De Bilt. The black line indicates a zero error and the red dashed line is the linear regression of the data points.} 
	\includegraphics[width= 0.5\textwidth]{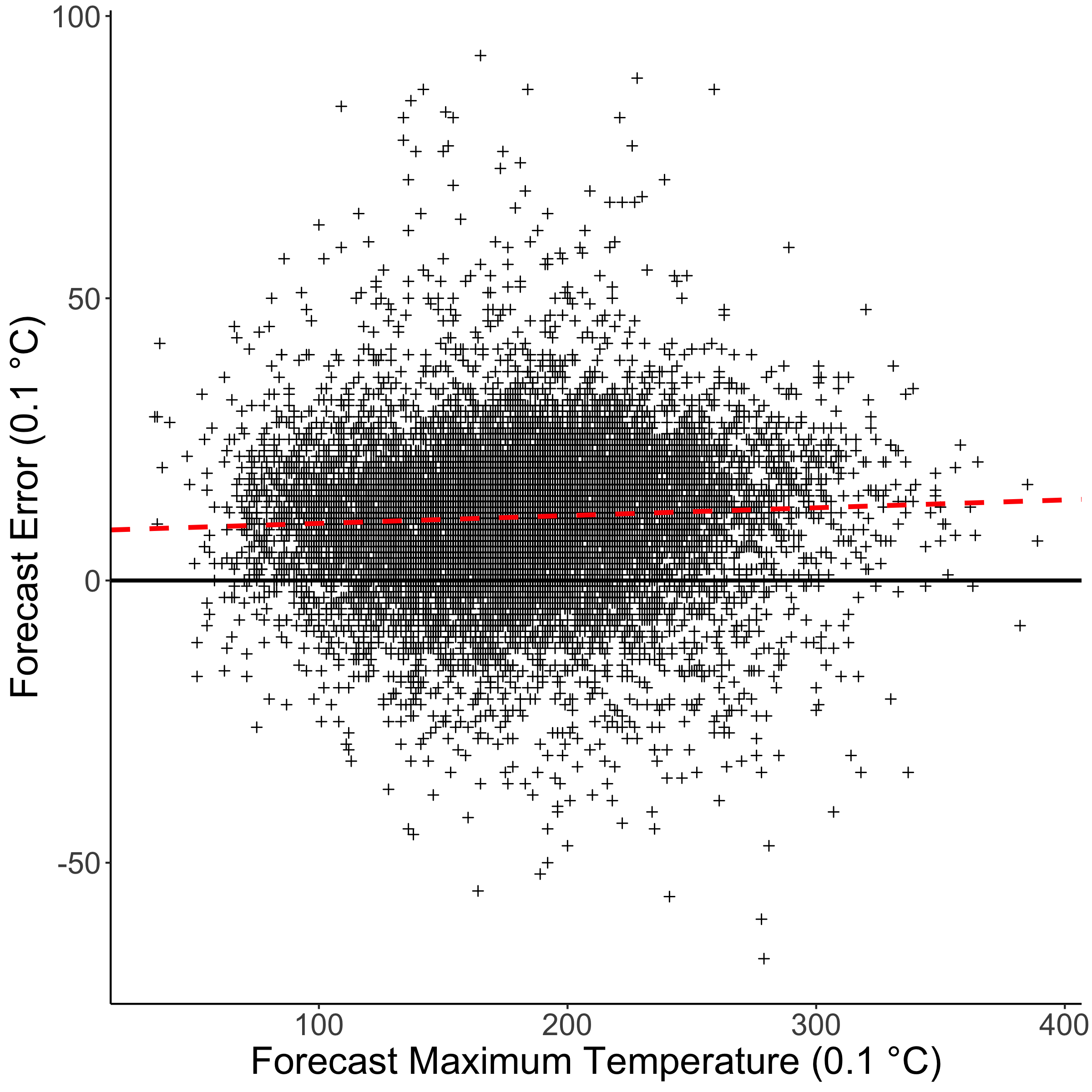}
\end{figure}

We use maximum temperature observed at seven stations spread across The Netherlands, namely Den Helder, Schiphol, De Bilt, Eelde, Twente, Vlissingen and Maastricht (\url{http://projects.knmi.nl/klimatologie/daggegevens/selectie.cgi}). The focus is on high temperatures, hence we consider only observations  from mid-April until mid-October. In total, we look at 9 years of data ranging from 2011 to 2019.

As covariates we use the output of the ECMWF model, which contains a 51 member ensemble and a higher resolution deterministic run. These forecasts are initiated two times a day, at 00 UTC and at 12 UTC, but here we use only forecasts of the latter run. We define the lead time of the forecast as the time difference between the start of the day for which the forecast is valid and the initiation time of the forecast. For this analysis we will consider forecasts with lead times equal to $36 + 24k$ hours for $k =0, 1,2,3,4,5$. The ensemble contains 51 exchangeable members and in order to use them we compute a set of summary statistics from the ensemble. These summary statistics are the mean, standard deviation, quantiles and number of ensemble members exceeding a pre-specified threshold. For the quantiles in our application we choose the 25, 50 and 75 percent quantiles. Thresholds are chosen as to extract different types of information from the ensemble relative to the weather phenomenon itself. For cloud cover we use three thresholds, 20 percent, 50 percent and 80 percent of cloud cover to create variables measuring probabilities of a few to no clouds, partly clouded weather and clouded weather.

Apart from the forecasts for maximum temperature and cloud cover, we consider other covariates including forecasts for daily average temperature at 2m, dew point temperature, minimum temperature, daily average wind speed and daily accumulated precipitation. For long lead times, predictability of these typical weather phenomena decreases, but the range of predictability of for example flow pattern at 500 hPa  extends much further. Therefore the first three principal components flow pattern at 500 hPa over Europe are also used as predictors \cite{Kruizinga1979}. Note that these covariates are the same for each station.

For the response variable we consider the forecast error, which we obtain by subtracting the deterministic forecast run from the observed maximum temperature. By doing so, the seasonality in the temperature is largely reduced. In Figure \ref{fig:bias}, the forecast error is clearly visible as the distance between the red linear regression line and the x-axis is rather large. Additionally it is clear that the spread of error changes as a function of the deterministic forecast. A possible explanation is that there is still remain seasonality effects that are not taken care of by a simple linear effect. Therefore, also the sine and cosine of the day of the year with a period of one year and half a year are included as two predictors. In total this gives us 71 covariates. For a given lead time an observation on a given day is denoted by $(Y,\bb X)$, with $Y$ the error of the deterministic run and $\bb X$ the 71 dimensional covariate vector. 


In this  section, we will explore 3 methods, quantile random forests as in \cite{Athey2019} with all variables, quantile random forests with variables selected by our forward variable selection and Non-homogeneous Gaussian Regression (NGR) \cite{Gneiting2005}. This third method is known in the meteorology literature as an EMOS (Ensemble Model Output Statistics) method and is used as a standard approach in post-processing. The NGR method assumes the data follow a Gaussian model, 
\begin{align*}
	&Y|\bb X= \bb x \sim N\left({\bb x}^T\beta,\exp({\bb x}^T\gamma)\right). 
\end{align*}
The parameters $\beta$ and $\gamma$ are then estimated by maximum likelihood. For this model, we select variables based on the Bayesian Information Criterion (BIC) by a forward and backward stepwise approach. 

For each station and lead time, we  fit a separate model. The models are estimated with a 9-fold cross validation, each time leaving out a single year.  In Figure \ref{fig:crps} the CRPS risk is shown as a function of lead time, where the box-plots contain the CRPS risk for all stations. Then in Figure \ref{fig:nr-cov} the number of selected variables is shown for our method and NGR, where we leave out the random forest with all 71 variables.

\begin{figure}%
    \centering
    \subfloat[CPRS]{ \label{fig:crps} {\includegraphics[width=0.45\textwidth]{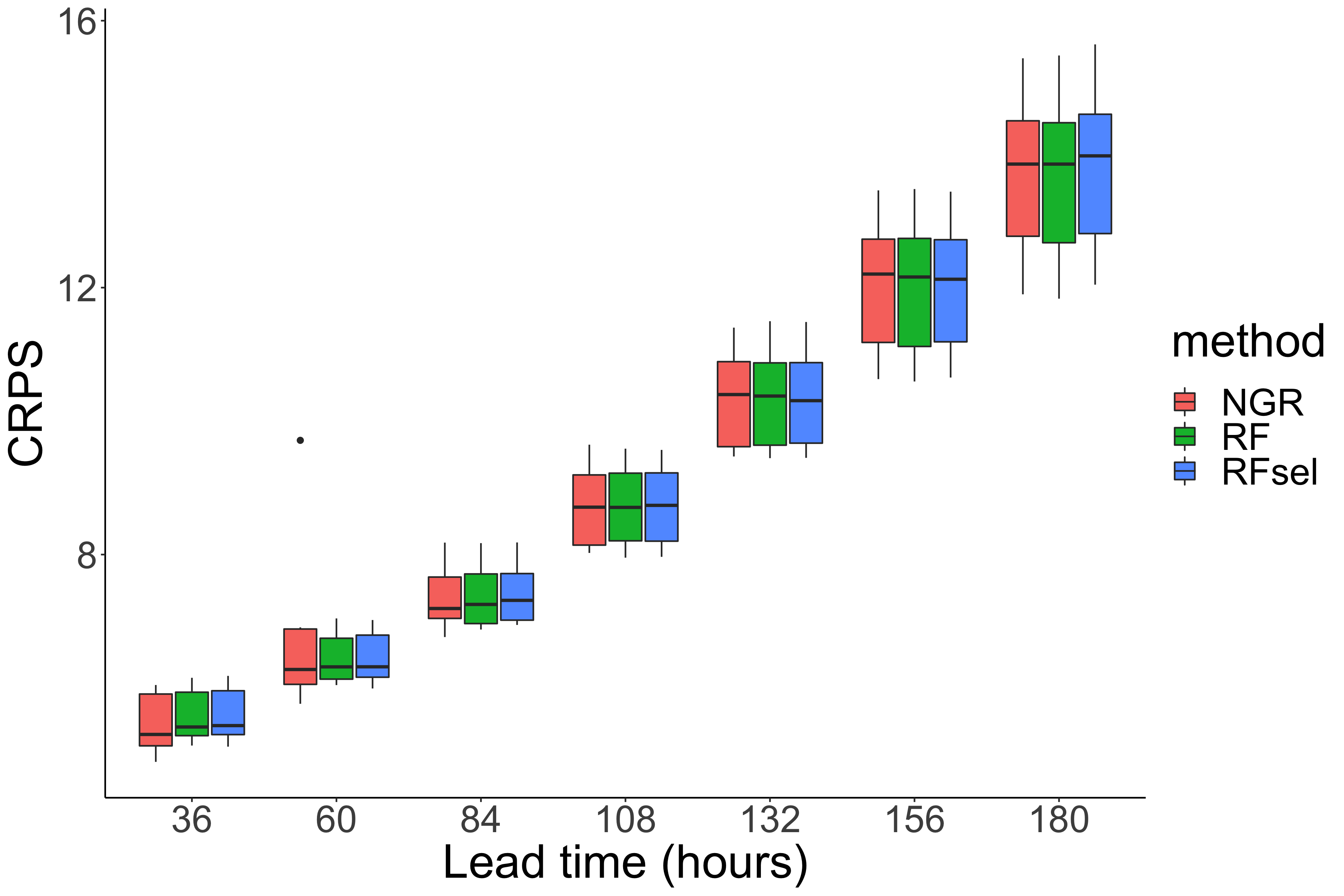} }}%
    \qquad
    \subfloat[Selected variables]{\label{fig:nr-cov} {\includegraphics[width=0.45\textwidth]{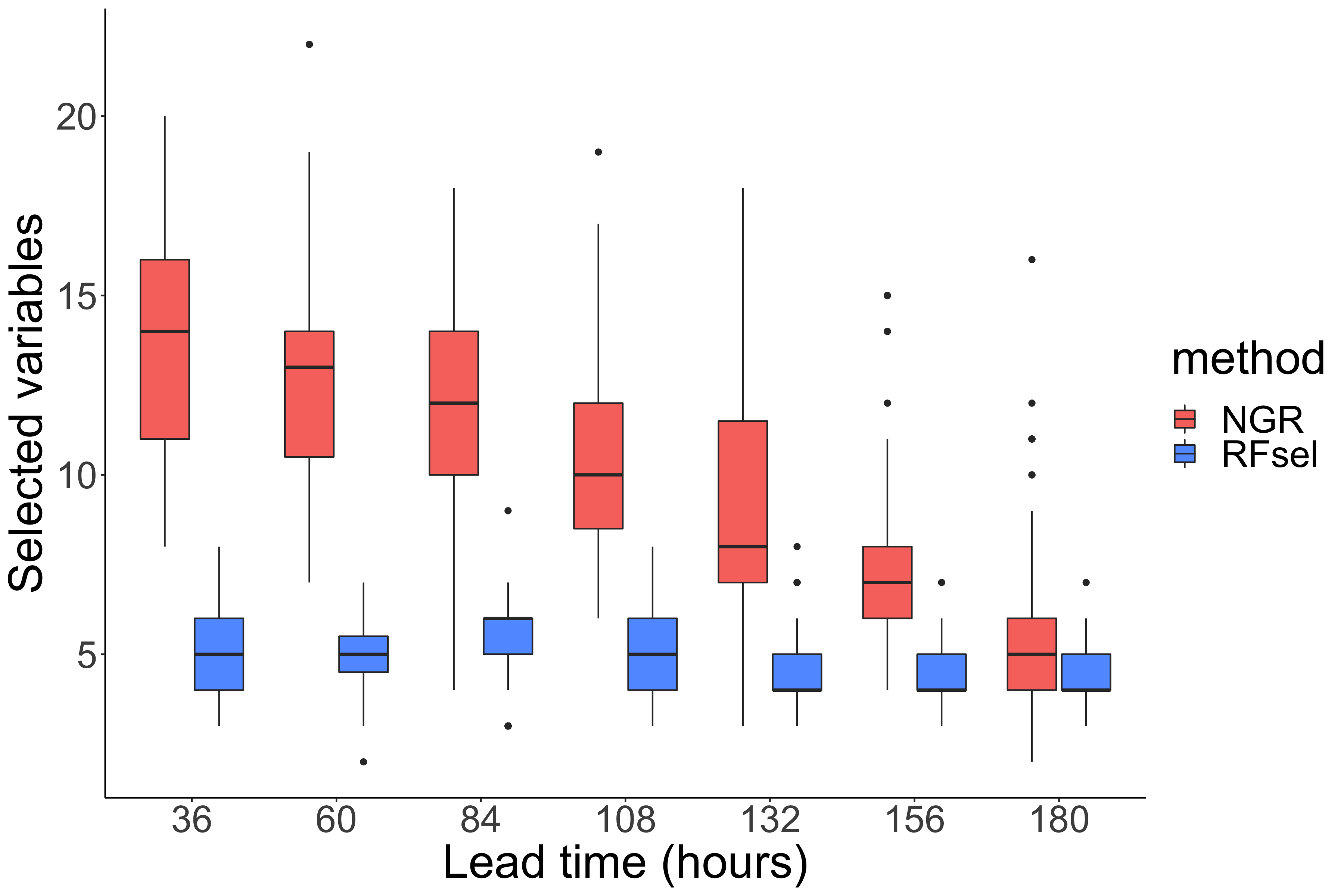} }}%
    \caption{(Left panel) the CRPS risk against leadtime where the boxplots contain the CRPS risk for each station. (Right panel) the number of selected variables against leadtime.}%
\end{figure}
 
Based on the CRPS, all methods perform comparably. This is also confirmed by other verification measures such as reliability diagrams, quantile reliability diagrams and probability integral transform histograms, which are not shown in this paper. A selection of these diagrams is shown in the appendix. The interesting part comes from the number of selected variables. Our method selects a small portion (less than $10\%)$ of covariates, substantially less than NGR.
 We  investigate this further by considering which variables are selected. The result is visualized in Figure \ref{fig:selected-variables}. For each lead time, the color indicates the frequency of a covariate being selected by 63 estimated models (7 locations and 9 cross-validation sets per location). An extremely important variable would be selected all 63 times.

\begin{figure}%
    \centering
    \includegraphics[width=\textwidth]{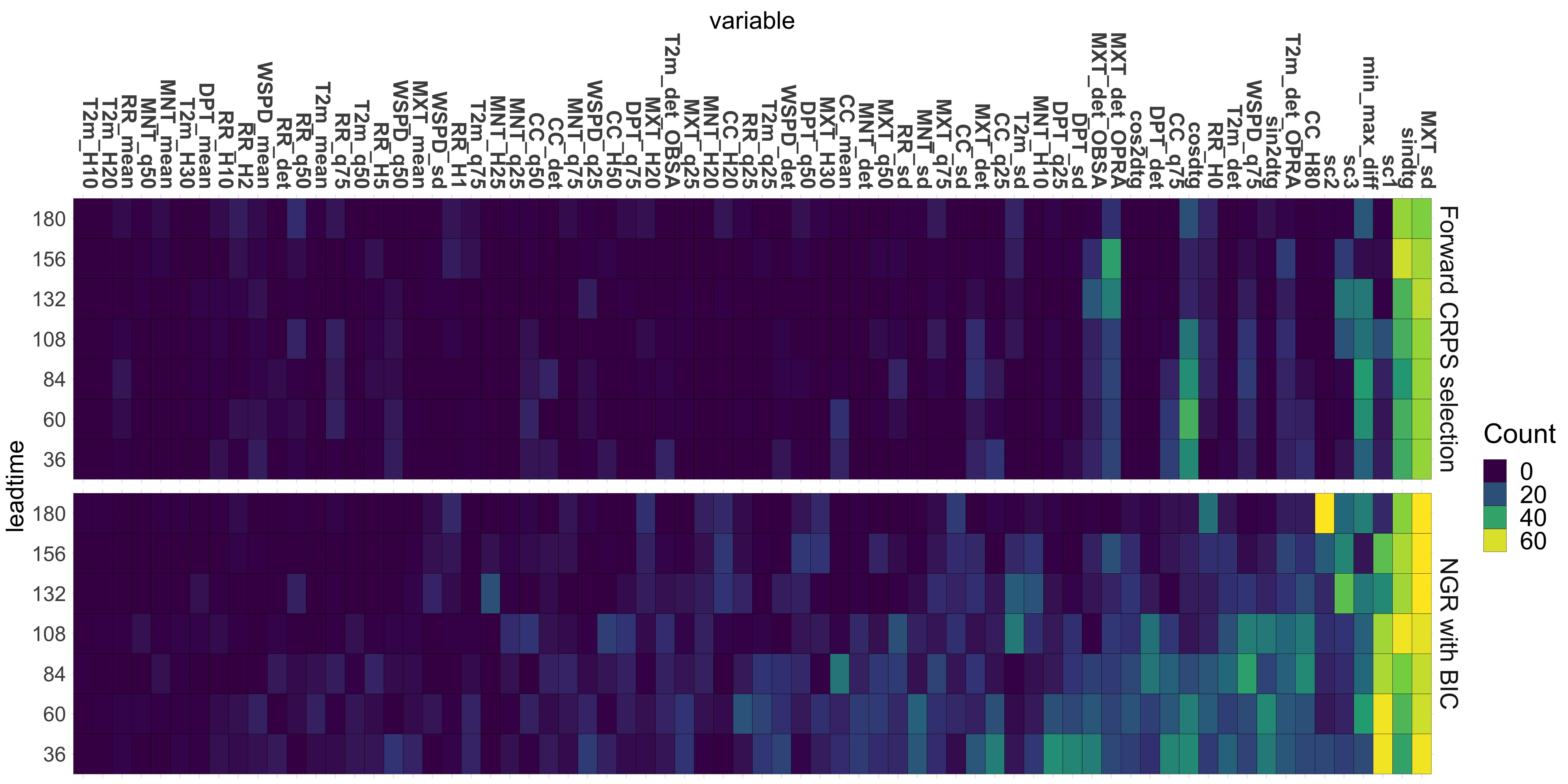}
    \caption{Selected variables for each lead time. All cross validations and all stations have been aggregated where the maximum number of times a variable can be selected is 63.}%
    \label{fig:selected-variables}%
\end{figure}

Yellow boxes correspond to a few variables that are always selected. But the number of light blue boxes is much smaller for our method compared to NGR. From this we  conclude that our method selects fewer variables and it also selects similar variables for different stations. This suggests that our variable selection method is more robust compared to the NGR method for short lead times, where  a diverse set of variables is selected.

\begin{figure}%
    \centering
    \includegraphics[width=0.8\textwidth]{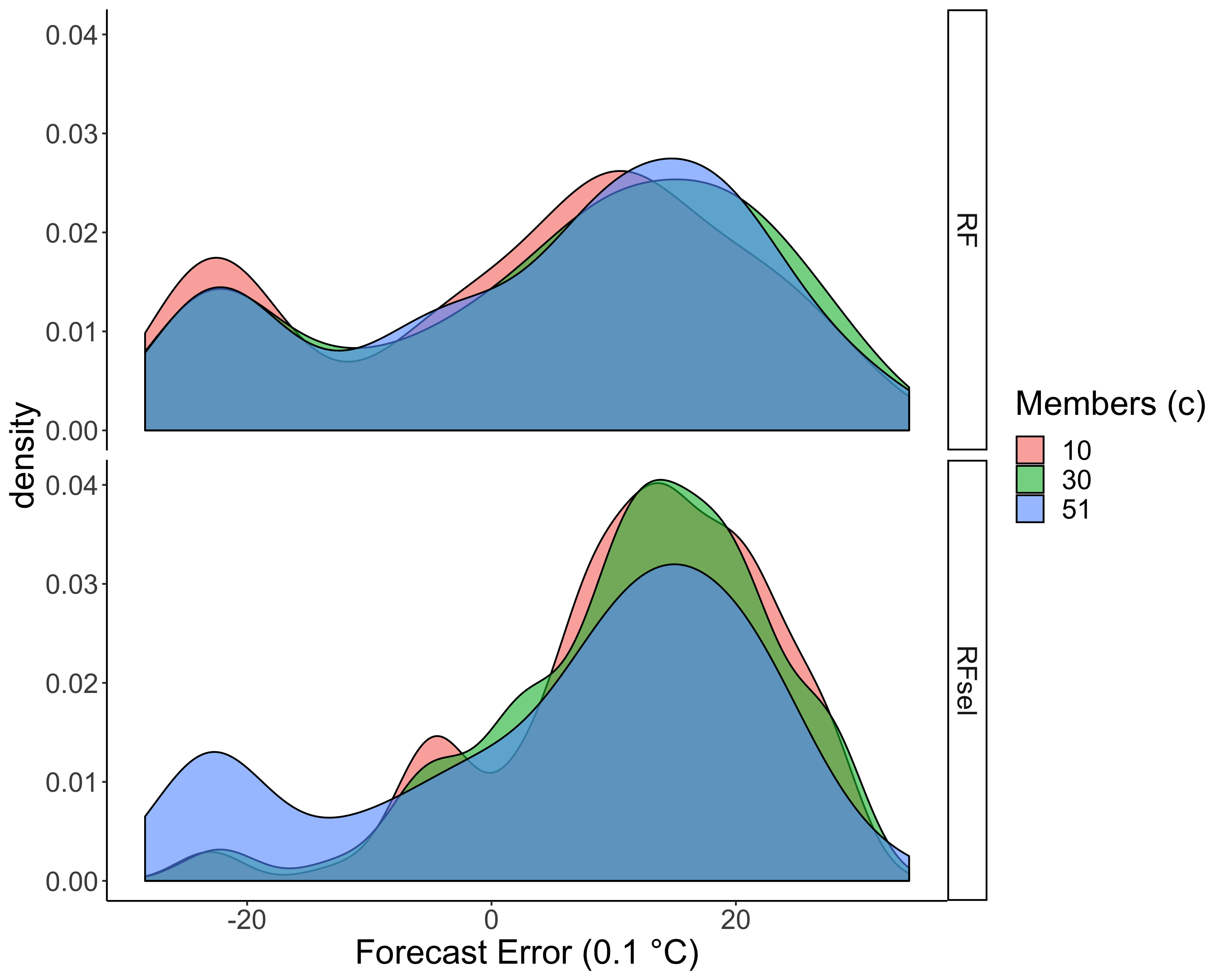} %
    \caption{The conditonal distribution of the forecast error, based on estimated models for De Bilt with lead time 36 hours. Different colors indicate different values of 
		cloud cover while the values of other covariates are fixed to be the same as that for 31-05-2018 at de Bilt.
		 Top figure for the random forest model with all variables and the bottom figure for the random forest model with the selection of variables.}%
    \label{fig:variable-effect}%
\end{figure}

The main variables that our method selects are the sine of the day of the year, the standard deviation of the ensemble forecast and variables related to cloud cover. Since our procedure typically selects a small set of variables, it is then feasible to interpret the estimated model. For instance, to investigate how a selected covariate, say $X_j$ influences the forecast distribution of $Y$,  one can compare the conditional distribution of $Y$ given different values of  $X_j$ while the other covariates denoted by $\bb X_{(-j)}$ are kept the same. We consider $Y$, the forecast error at de Bilt with lead time 36 hours and $X_j$ the cloud cover, which is the number of ensemble members with cloud cover exceeding $50\%$.  The values of other covariates are fixed the same as the data of 31-05-2018 at De Bilt, denoted by $\bb X_{(-j)}=\bb x^*_{(-j)}$. Figure \ref{fig:variable-effect} shows the conditional density of $Y$ given $(X_j=c, \bb X_{(-j)}=\bb x^*_{(-j)})$, where different colors indicate three different values of $c$. Note that all 51 ensemble members exceed 50\% cloud cover. As shown in the lower panel of Figure \ref{fig:variable-effect}, cloud cover clearly has an effect based on the estimation of our method: $c=51$ yields a bimodal distribution while $c=10$ leads to a unimodal distribution. This suggests that in this configuration, higher cloud cover implies a higher chance for a negative forecast error (left mode in the plot). However, the distributions obtained by random forest (without variable selection) are very similar; see the upper panel of the figure. This is because that there are other covariates correlated to cloud cover, and these covariates still indicate that there is a high cloud cover even when the number of ensemble members exceeding 50\% could cover is set to 10.  In other words, changing the value of a single variable in a random forest with many correlated covariates is {\it not} interpretable. Such a random forest model fails to capture the effect of a signal variable.


\section{Summary and discussion}
\label{sec:dis}
In this paper, we have proposed a general framework for a forward variable selection with respect to a loss function. We show in population sense that under an independence assumption between covariates and by choosing the continuous ranked probability score as loss function that the forward selected variables form the correct set with respect to the CRPS risk functional. Applying the method in a random forest set-up, we show that the out-of-bag samples can be efficiently used to asses predictive performance. The main difficulty in the procedure is determining the stopping time, that is when selecting more variables does not add in predictive performance. Due to randomness and the inherent greedy variable selection procedure in the random forest algorithm this can not be determined by the calculated predictive performance.  Instead in a single forward selection step we use the predictive performance of each possible set to construct a test to detect increasing predictive performance. The procedure then stops a null hypothesis of non increasing predictive performance can not be rejected. We show that this test is consistent.

With a simple simulation study we show that our variable selection method, compared to a backward selection based on a permutation importance measure, is more capable of discriminating  between signal variables and noise variables. This improvement is shown for various sample sizes and correlations between the covariates.

In an application on post-processing maximum temperature, our method shows consistency in the number of selected variables and in the variables being selected over several stations. Moreover, our method selects less than 10 percent of the covariates and still attains the same predictive power as the quantile random forest with all covariates. Further, it is easier to interpret our resulting model, due to the largely reduced number of covariates. Without variable selection, it is hardly possible to analyse the effect of a single covariate in a random forest model when it is heavily correlated to other covariates. In our data example, in the presence of thick cloud cover, our random forest model indicates that there is a higher risk of over forecasting (lower panel of Figure \ref{fig:variable-effect}) instead of under-forecasting which was indicated by Figure \ref{fig:bias}. 

There are two interesting directions for future research. First, the theoretical results in Sections \ref{sec:mathematical-set-up} and \ref{sec:method} are derived under the assumption that the covariates are independent. However, the ability of our method to select signal variables from a correlated setting is evidenced by our simulation study and data application. It is interesting to investigate such a setting.  Second, we focus in this paper on how this forward method behaves for the CRPS, but the mathematical set-up in Section \ref{sec:mathematical-set-up}  is much more general and allows to select variables with respect to other loss functions. It would be interesting to extend the current results to a more general set of loss functions.

\bibliographystyle{plain}
\bibliography{bibliography}

\section*{Acknowledgements}
We would like to thank Maurice Schmeits, Kiri Whan and Dirk Wolters for many useful discussions on the application of maximum temperature.

This work is part of the research project ``Probabilistic forecasts of extreme weather utilizing advanced methods from extreme value theory" with project number 14612 which is financed by the Netherlands Organisation for Scientific Research (NWO).

\appendix
\section{CRPS calculations}

Here we show for an observation $y$ and a distribution function $F$ that the CRPS calculated from the quantile perspective as well as from the distribution function perspective are equivalent as shown in \cite{Laio2007}, i.e we show that 
\begin{equation}
	2 \int_0^1 \rho_{\tau}(y - F^{-1}(\tau)) \mbox{d}\tau = \int_{-\infty}^{\infty} (I(y\leq z) - F(z))^2 \mbox{d}z. 
\end{equation}

We have,
\begin{align*}
	2 \int_0^1 \rho_{\tau}(y - F^{-1}(\tau)) \mbox{d}\tau &= 2 \int_0^1 (I(y \leq F^{-1}(\tau)) - \tau)(F^{-1}(\tau) - y) \mbox{d}\tau\\
	&= 2 \int_{-\infty}^{\infty} (I(y \leq z) - F(z))(z - y)f(z) \mbox{d}z\\
	&= \left.-(I(y \leq z) - F(z))^2(z-y)\right|_{-\infty}^{\infty} + \\
	& \int_{-\infty}^{\infty} (I(y\leq z) - F(z))^2 \mbox{d}z\\
	&= \int_{-\infty}^{\infty} (I(y\leq z) - F(z))^2 \mbox{d}z
\end{align*}

Here we use a substitution in the second line of $\tau = F(z)$ and in the third line we apply integration by parts.

\section{Calibration of forecasts for lead time 60 and station De Bilt}

 Figure \ref{fig:rank-hist} shows a histogram of the $\hat{F}(Y)$ where $\hat{F}$ is the forecast distribution for observation $Y$. If $F$ is calibrated the histogram should look like the histogram based on standard uniform random variable. 
\begin{figure}[ht]%
    \centering
    \includegraphics[width=0.8\textwidth]{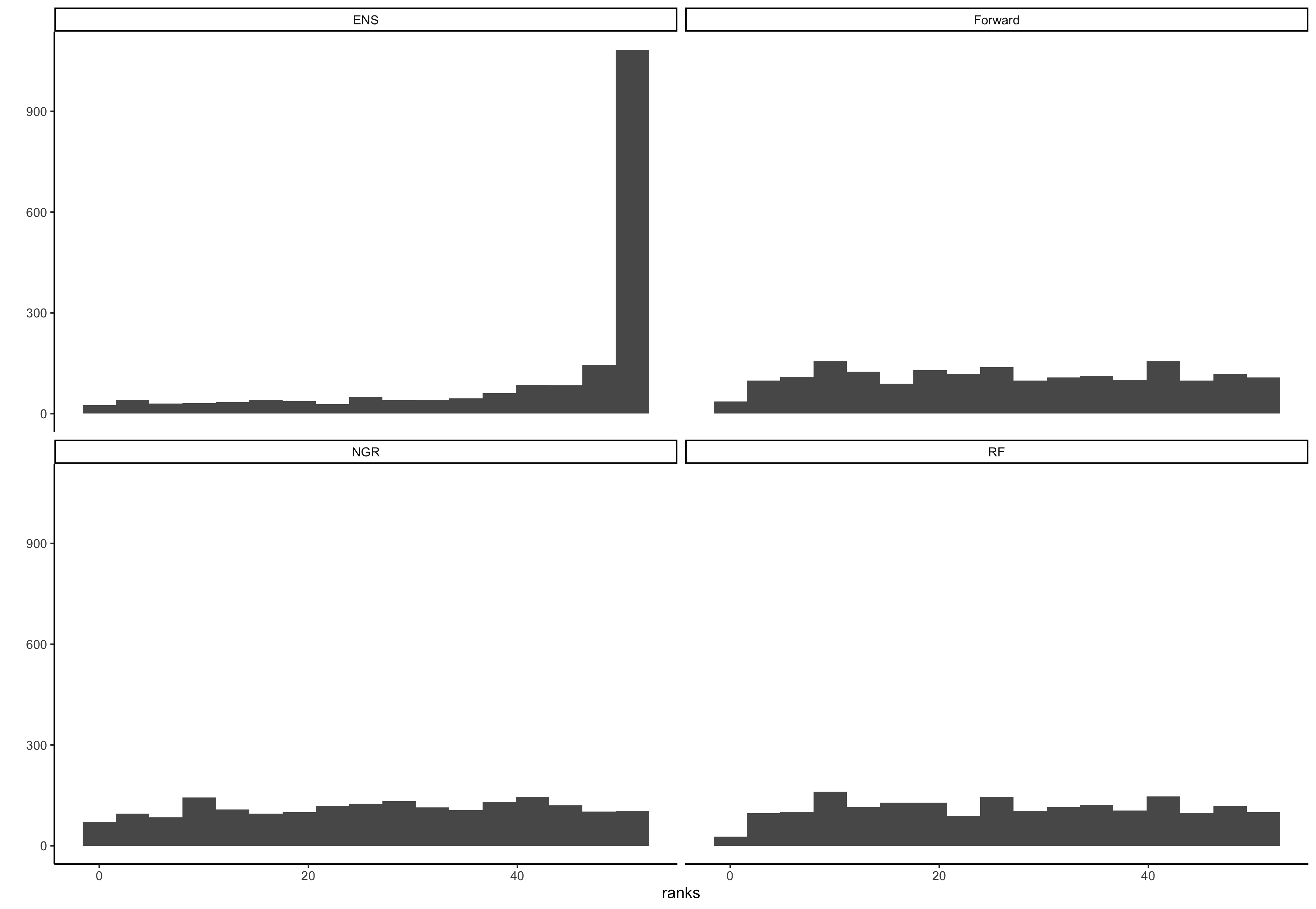}
    \caption{Rank histograms for lead time 60 h and station De Bilt. For forward selection, NGR, random forest and the raw ensemble forecast}%
    \label{fig:rank-hist}%
\end{figure}

Figure \ref{fig:RD} shows reliability diagrams. Let $t$ be a threshold and  define $p = \hat{F}(t)$ and $I = I(Y\leq t)$ for each forecast. A reliability diagram bins the probabilities $p$ in equally sized bins. The average indicator $I$ should be the same as the average $p$. Hence plotting these averages they should be approximately on the identity line; for detailed explanation we refer to \cite{Wilks2011}. 

\begin{figure}[ht]%
    \centering
    \includegraphics[width=0.8\textwidth]{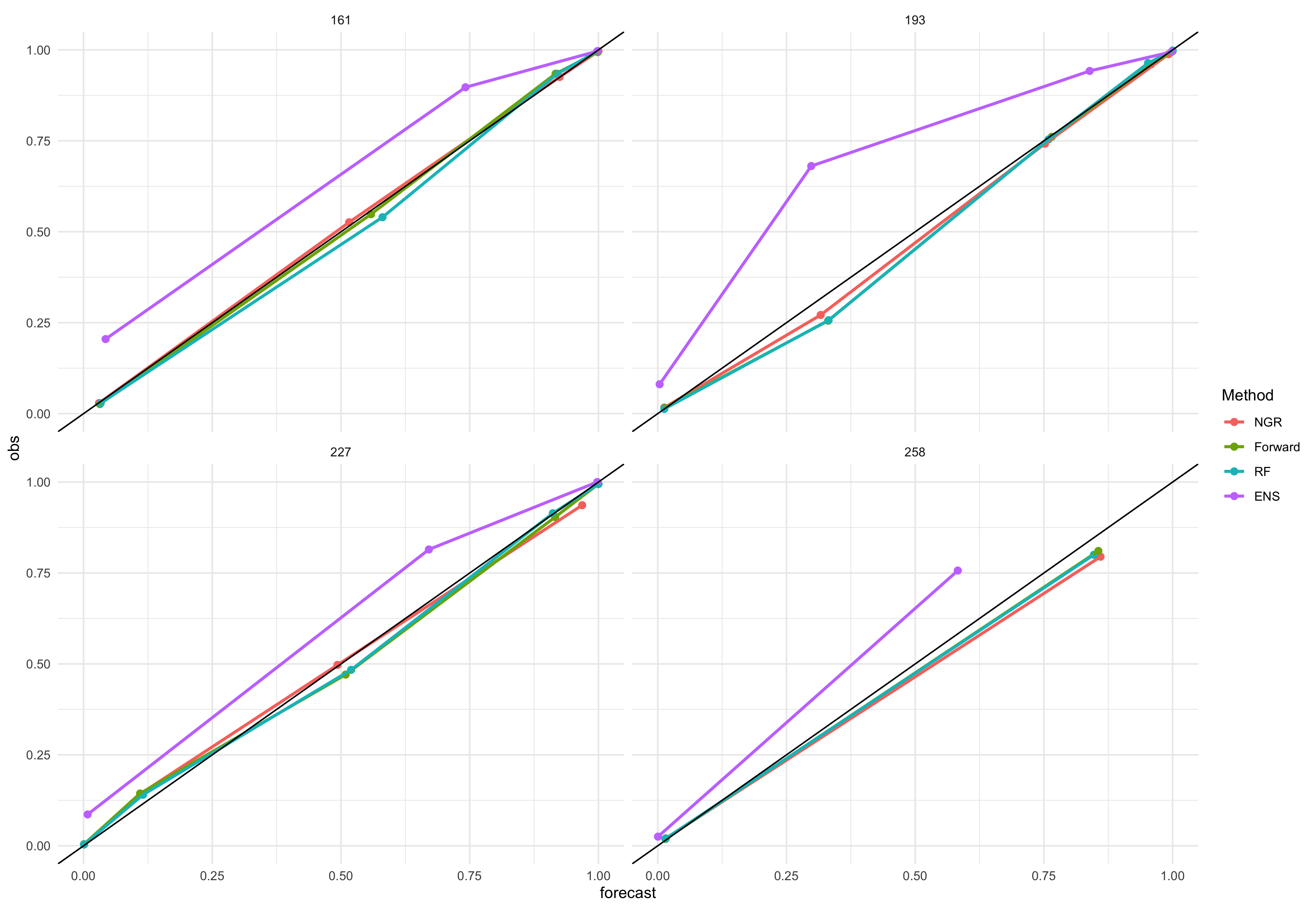}
    \caption{Reliability diagrams with thresholds equal to the $0.25,0.5,0.75$ and $0.9$ observational quantiles for lead time 60 and station De Bilt. Methods compared are: forward selection, NGR, random forest and the raw ensemble forecast}%
    \label{fig:RD}%
\end{figure}

Figure \ref{fig:QRD} shows quantile reliability diagrams. Let $\tau$ be a probability level and $\hat{Q}$ the forecast quantile function. Define $q = \hat{Q}(\tau)$ for each forecast. A quantile reliability diagram bins the quantiles $q$ in equally sized bins. The $\tau$ quantile of observation $Y$ should be the same as the average $q$. Hence plotting these against each other should be approximately on the identity line; for detailed explanation we refer to \cite{Bentzien2014}.

\begin{figure}[ht]%
    \centering
    \includegraphics[width=0.8\textwidth]{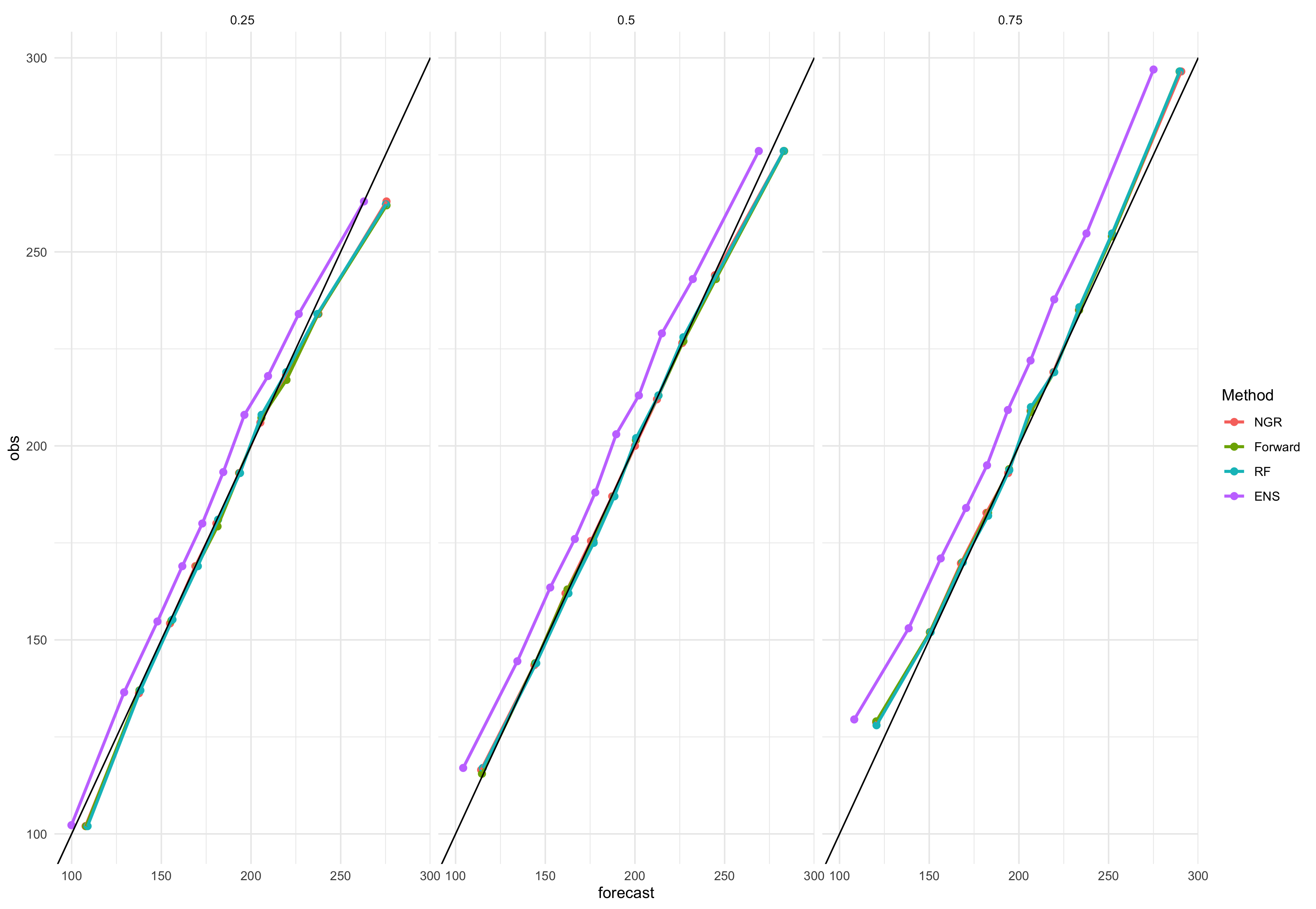}
    \caption{Quantile reliability diagrams for quantile levels equal to $0.25,0.5$ and $0.75$ for lead time 60 and station De Bilt. Methods compared are: forward selection, NGR, random forest and the raw ensemble forecast}%
    \label{fig:QRD}%
\end{figure}

\end{document}